\documentclass[12pt]{article}
\usepackage[utf8]{inputenc}
\usepackage[T1]{fontenc}

\usepackage[margin=3cm]{geometry}

\usepackage{amsthm, amsmath, amssymb}

\usepackage{thmtools}
\usepackage{thm-restate}

\usepackage{hyperref}
\usepackage{graphicx}

\usepackage{tikz}
\usetikzlibrary{calc}
\usetikzlibrary{decorations.pathreplacing,calligraphy}

\usepackage{subcaption}
\usepackage{ifthen}

\hypersetup{
  pdftitle = {Long induced paths in minor-closed graph classes and beyond},
  pdfauthor = {C. Hilaire and J.-F. Raymond},
  colorlinks = true,
  linkcolor = black!30!red,
  citecolor = black!30!green
}

\newtheorem{theorem}{Theorem}[section]

\newtheorem{lemma}[theorem]{Lemma}
\newtheorem*{lemma*}{Lemma}
\newtheorem{corollary}[theorem]{Corollary}
\newtheorem{conjecture}{Conjecture}[section]

\theoremstyle{remark}
\newtheorem{remark}{Remark}[section]

\newcommand{\N}{\mathbb{N}}
\newcommand{\R}{\mathbb{R}}
\newcommand{\cR}{\mathcal{R}}
\newcommand{\T}{\mathcal{T}}
\newcommand{\PR}{\mathrm{PR}}
\newcommand{\TR}{\mathrm{TR}}

\DeclareMathOperator{\pw}{\mathbf{pw}}
\DeclareMathOperator{\tw}{\mathbf{tw}}

\DeclareMathOperator{\adh}{adh}

\newcommand{\floor}[1]{\left\lfloor#1\right\rfloor}
\newcommand{\ceil}[1]{\left\lceil#1\right\rceil}
\newcommand{\intv}[2]{\left \{ #1, \dots, #2 \right \}}

\newcommand{\lpp}{(\hyperlink{it:lpp}{$\star$})} 
\newcommand{\hpp}{(\hyperlink{it:hpp}{$\star\star$})} 

\title{Long induced paths\\in minor-closed graph classes and beyond}
\author{Claire Hilaire\thanks{Univ. Bordeaux, CNRS, Bordeaux INP, LaBRI, UMR 5800, F-33400, Talence, France.} \and Jean-Florent Raymond\thanks{Université Clermont-Auvergne, CNRS, LIMOS, 63000 Clermont-Ferrand, France. This author was supported by the ANR project GRALMECO (ANR-21-CE48-0004).}}
\date{}

\begin{document}

\maketitle

\begin{abstract}
  In this paper we show that every graph of pathwidth less than~$k$ that has a path of order $n$ also has an induced path of order at least $\frac{1}{3} n^{1/k}$.
  This is an exponential improvement and a generalization of the polylogarithmic bounds obtained by Esperet, Lemoine and Maffray (2016) for interval graphs of bounded clique number. We complement this result with an upper-bound.
  
  This result is then used to prove the two following generalizations:
  \begin{itemize}
  \item every graph of treewidth less than $k$ that has a path of order $n$ contains an induced path of order at least $\frac{1}{4} (\log n)^{1/k}$;
  \item for every non-trivial graph class that is closed under topological minors there is a constant $d \in (0,1)$ such that every graph from this class that has a path of order $n$ contains an induced path of order at least $(\log n)^d$.
  \end{itemize}
  We also describe consequences of these results beyond graph classes that are closed under topological minors.
\end{abstract}

\section{Introduction}\label{sec:intro}

In this paper we are concerned with the dependency between the maximum lengths of paths and induced paths in graphs.
Every induced path is a path. Conversely, does the existence of a long path in a graph imply that of a long induced path? This question may be formalized as asking for the existence of an increasing function $f$ such that the following property holds:
\begin{enumerate}
\item[\hypertarget{it:lpp}{$(\star)$}] \textit{For every $n\in \N$, if a graph has a path of order $n$, then it has an induced path of order at least~$f(n)$.}
\end{enumerate}
In general, such a function does not exist, as shown by cliques and bicliques.
In 1982, Galvin, Rival, and Sands showed that these are the only obstructions by providing an increasing function satisfying \lpp{} for graphs excluding a biclique as a subgraph~\cite[Theorem~4]{GALVIN19827}.
Their proof relies on the infinite Ramsey's theorem for 4-tuples, hence their result should first and foremost be seen as an existential result, rather than one providing accurate and tight bounds.

A better lower-bound was given in \cite[Lemma~6.4]{nevsetvril2012sparsity} by Ne{\v{s}}et{\v{r}}il and Ossona de Mendez in the case of $k$-degenerate\footnote{For a positive integer $k$, a graph $G$ is said to be $k$-degenerate if every subgraph of $G$\- (including $G$ itself) has a vertex of degree at most~$k$.} graphs: there, a path of order $n$ implies the existence of an induced path of order at least $\frac{\log \log n}{\log (k+1)}$.
They also proposed as an open problem to find the maximum function $f\colon \N^2 \to \N$ such that, for every $k\in \N$, property \lpp{} holds for $k$-degenerate graphs with bound $n \mapsto f(k,n)$ \cite[Problem~6.1]{nevsetvril2012sparsity}.

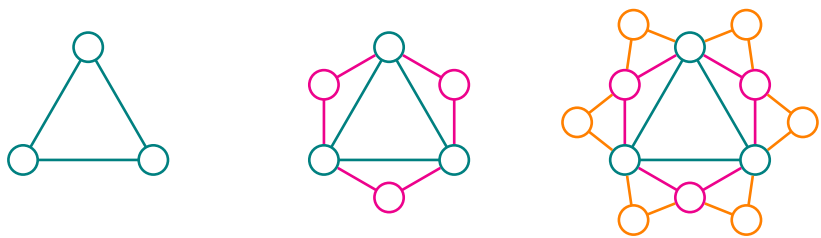
\begin{figure}[h]
\centering
\begin{tikzpicture}[every node/.style = {draw,fill=white,circle}, every path/.style={line width=1pt}]
\begin{scope}[rotate = 90]
    \draw[color=teal] (0:1) node {} --
    (120:1)  node {} --
    (-120:1)  node {} -- cycle;
\end{scope}
\begin{scope}[xshift = 4cm, rotate = 90]
    \draw[color=teal] (0:1) node (a0) {} --
    (120:1)  node (a1) {} --
    (-120:1)  node (a2) {} -- cycle;
    \draw[color=magenta] (a0) -- (60:1) node (b0) {} --
    (a1) -- (180:1) node (b2) {} --
    (a2) -- (-60:1) node (b2) {} -- (a0);
\end{scope}
\begin{scope}[xshift = 8cm, rotate = 90]
    \draw[color=teal] (0:1) node (a0) {} --
    (120:1)  node (a1) {} --
    (-120:1)  node (a2) {} -- cycle;
    \draw[color=magenta] (a0) -- (60:1) node (b0) {} --
    (a1) -- (180:1) node (b1) {} --
    (a2) -- (-60:1) node (b2) {} -- (a0);
    \draw[color=orange] (a0) -- (30:1.5) node {} --
    (b0) -- (90:1.5) node {} --
    (a1) -- (150:1.5) node {} --
    (b1) -- (210:1.5) node {} --
    (a2) -- (270:1.5) node {} --
    (b2) -- (330:1.5) node {} --
    (a0);
\end{scope}

\end{tikzpicture}
\caption{First terms of a sequence of outerplanar graphs where induced paths have length at most logarithmic in the order of the graph.}
\label{fig:op}
\end{figure}

A graph is said to be \emph{outerplanar} if it can be drawn in the plane with no edge crossing and with all the vertices on the outer face.
Outerplanar graphs form a subclass of planar graphs and are known to be 2-degenerate.
Already within this simple class there are infinite families of graphs where the order of the longest induced path is at most logarithmic in the order of the graph \cite{arocha2000long}, see Figure~\ref{fig:op} for an example.
This motivates the research for classes of degenerate graphs where a polylogarithmic\footnote{That is, of the form $n \mapsto c(\log n)^d$ for some reals $c,d>0$.} lower-bound for the function of property~\lpp{} could be obtained. As a possible answer to the aforementioned question of Ne{\v{s}}et{\v{r}}il and Ossona de Mendez, it was conjectured in \cite{esperet2017long} that this would be the \emph{correct} bound for $k$-degenerate graphs.

\begin{conjecture}[{\cite[Conjecture~1.1]{esperet2017long}}]\label{conj:esp}
For every integer $k$ there is a constant $d$ such that every $k$-degenerate graph that has a path of order $n$ also has an induced path of order at least $(\log n)^d$.
\end{conjecture}

In 2000, Arocha and Valencia \cite{arocha2000long} and later Di Giacomo, Liotta and Mchedlidze \cite{DIGIACOMO201647} considered 3-connected planar graphs and {2-connected} outerplanar graphs.

Their results have then been improved by Esperet, Lemoine, and Maffray who showed that 2-connected outerplanar graphs and 3-connected planar graphs of order $n$ have an induced path of order $\Omega(\log n)$~\cite[Theorems 3.2 and 3.5]{esperet2017long}.\footnote{These results do not require a long path to exist in the considered graph because this follows from the connectivity requirement in the considered graph classes, see e.g.~\cite{CHEN200280}.}
They used these results to prove the following statement.

\begin{theorem}[{\cite[Theorem~3.8]{esperet2017long}}]\label{th:espgenus}
For every $g\in \N$ there is a constant $c = \frac{1}{2\sqrt{6}} - o(1)$ such that for every graph $G$ embeddable in a surface with Euler genus at most $g$, if $G$ has a path of order $n$, then $G$ has an induced path of order at least $c \sqrt{\log n}$.
\end{theorem}

Besides planar graphs, a prominent type of degenerate graphs consists of graphs of bounded treewidth (to be defined in Section~\ref{sec:prelim}). Indeed, every graph of treewidth at most $k$ is $k$-degenerate.
If $G$ has treewidth $k$ and any addition of an edge between two vertices of $G$ yields a graph of treewidth larger than $k$, then $G$ is called a \emph{$k$-tree}.
In the same paper, Esperet et al. obtained the following bound for $k$-trees. 

\begin{theorem}[{\cite[Theorem~2.2]{esperet2017long}}]\label{th:ktree}
For every $k$, if a $k$-tree has a path of order $n$ then it has an induced path of order at least $\frac{\log n}{k \log k}$.
\end{theorem}

Using similar ideas they were able to show the following.

\begin{theorem}[{\cite[Corollary~2.5]{esperet2017long}}]
If a graph $G$ of treewidth at most 2 has a path of order $n$ then it has an induced path of order at least $\left (\frac{1}{2} - o(1)\right ) \log n$.
\end{theorem}

However it is not clear for larger values of $k$ how the ideas of the proof of Theorem~\ref{th:ktree} could be adapted to the more general setting of graphs of treewidth less than $k$. 
Actually the upper-bound of $(k+1)(\log n)^{\frac{2}{k-1}}$ on the function of property \lpp{} for graphs of treewidth less than $k$ provided in the same paper even suggests that new techniques will be needed to approach such graphs.

The second type of class of graphs of bounded treewidth considered in the same paper is interval graphs of bounded clique number.\footnote{The \emph{clique number} of a graph is the maximum order of a clique it contains. An interval graph of clique number $k$ has treewidth less than $k$.}

\begin{theorem}[{\cite[Theorem~4.1]{esperet2017long}}]\label{th:esp}
For every integer $k$ there is a constant $c$ such that, if $G$ is an interval graph with clique number at most $k$ and that has a path of order $n$, then $G$ has an induced path of order at least $c 
(\log n)^{\frac{1}{(k-1)^2}}$.
\end{theorem}

\subsection{Our contribution}

Our first result is a generalization of Theorem~\ref{th:esp} to graphs of bounded pathwidth\footnote{An interval graph has clique number $k$ iff it has pathwidth $k-1$.}, with an exponential improvement of the bound.

\begin{restatable}{theorem}{pwintro}
\label{th:pwintro}
For every $k \in \N$, if $G$ is a graph of pathwidth less than $k$ that has a path of order $n$, then $G$ has an induced path of order at least $\frac{1}{3}n^{1/k}$.
\end{restatable}

Its proof consists in a simple win/win strategy. We identify an induced path whose removal decreases the pathwidth of~$G$. Then either this path is at least as long as the bound promised by the statement and we are done, or its removal decreases the pathwidth without decreasing the number of vertices much, and we conclude by applying the induction hypothesis. Theorem~\ref{th:pwintro} is complemented by an upper-bound (that even holds for interval graphs) of $n^{2/k} + 1$ (Theorem~\ref{th:pwlb}), showing that the exponential dependency in $1/k$ in our lower-bound above is unavoidable.

We then show that in a graph of small treewidth that has a large path there is always (as a contraction) a graph of small pathwidth that has a long path. This statement is used to obtain the following polylogarithmic bound for graphs of bounded treewidth.

\begin{restatable}{theorem}{twintro}
\label{th:twintro}
For every $k \in \N$, if $G$ is a graph of treewidth less than $k$ that has a path of order $n$, then $G$ has an induced path of order at least $\frac{1}{4}(\log n)^{1/k}$. 
\end{restatable}
As mentioned above, Esperet et al. \cite{esperet2017long} constructed chordal graphs of clique number~$k$ that have a path of order $n$ and where no induced path has order more than $(k+1) (\log n)^{\frac{2}{k-1}}$. Therefore neither the logarithmic dependency in $n$ nor the exponential dependency in $1/k$ could be improved in our lower-bound above.

The operation of \emph{subdividing} an edge $uv$ of a graph consists in deleting the edge and adding a new vertex of degree 2 adjacent to $u$ and $v$.
A \emph{subdivision} of a graph $H$ is a graph obtained by applying a (possibly empty) sequence of edge subdivisions to~$H$. If a graph $G$ has a subgraph that is (isomorphic to) a subdivision of $H$, we say that $H$ is a \emph{topological minor} of~$G$.
A graph class is said to be \emph{closed under taking topological minors} if every topological minor of a graph of the class also belongs to the class.
It is \emph{non-trivial} if it is not the class of all graphs.

Non-trivial graph classes that are closed under topological minors form a very general setting that encompasses several fundamental graph classes such as any minor-closed graph class (e.g. planar graphs or more generally graphs of bounded Euler genus, and graphs of bounded pathwidth or treewidth), any immersion-closed graph class (like graphs of bounded cutwidth, carving width, or tree-cut width), and graphs of bounded degree.

The ideas developed in the proofs of Theorem~\ref{th:pwintro} and Theorem~\ref{th:twintro} allowed us to generalize their statement to this much more general setting.

\begin{theorem}\label{th:minclo}
For every non-trivial graph class $\mathcal{G}$ that is closed under taking topological minors there is a constant $d \in (0,1)$ such that if a graph $G \in \mathcal{G}$ has a path of order $n$, then $G$ has an induced path of order at least $(\log n)^d$.
\end{theorem}

The proof deals separately with the different parts that compose graphs excluding a topological minor (as given by the structure theorem of Grohe and Marx \cite{grohe2015structure}) and then shows how they can be combined together.
We actually prove a stronger statement than Theorem~\ref{th:minclo}, namely that the same outcome holds for all graphs that admit a tree decomposition where every torso is either almost embeddable in a surface of bounded genus or has almost bounded degree (see Theorem~\ref{th:master} for the formal statement).

Before proceeding to the proofs, let us first mention some consequences of our results beyond graph classes that are closed under topological minors.

\subsection{Consequences of our results}

\paragraph{Graphs of bounded chordality.}

Gartland, Lokshtanov, Pilipczuk, Pilipczuk, and Rzą{\.z}ewski proved in \cite{gartland2021finding} that in graphs of bounded chordality (i.e. graphs with no long induced cycles), degeneracy and treewidth are tied in the following sense.

\begin{theorem}[{\cite[Theorem~1.3]{gartland2021finding}}]
For every $k,\ell \in \N$ there is a constant $t$ such that if a graph $G$ has degeneracy at most $k$ and no induced cycle on at least $\ell$ vertices then $G$ has treewidth at most~$t$.
\end{theorem}

Together with the above result, Theorem~\ref{th:twintro} implies Conjecture~\ref{conj:esp} for graphs of bounded chordality.

\begin{corollary}
For every $k,\ell \in \N$ there is a constant $d \in (0,1)$ such that if a graph $G$ is $k$-degenerate with no induced cycle of order at least $\ell$ and $G$ has a path of order $n$, then $G$ has an induced path of order at least $(\log n)^d$.
\end{corollary}

\paragraph{Graphs of bounded cliquewidth.}
Cliquewidth is a graph parameter that is more general than treewidth.\footnote{Cliquewidth is more general in the sense that every class of graphs of bounded treewidth has bounded cliquewidth, while the converse does not hold in general.
Cliquewidth only appears in the statement of Corollary~\ref{cor:cliwi} so we refrain from giving its quite lengthy definition here and instead refer the interested reader to one of the papers on the topic such as~\cite{COURCELLE200077}.}
However, if we forbid arbitrarily large cliques and bicliques then each of the two parameters can be upper-bounded by a function of the other~\cite{COURCELLE200077}. Specifically, any $K_{s,s}$-subgraph free graph of cliquewidth less than $k$ has treewidth less than $3(k-1)(s-1)$, as proved in~\cite[Corollary~1]{10.1007/3-540-40064-8_19}. Thanks to this property we have the following consequence of Theorem~\ref{th:twintro}.

\begin{corollary}\label{cor:cliwi}
Let $k,s \in \N^+$, if $G$ is a $K_{s,s}$-subgraph free graph of cliquewidth less than $k$ that has a path of order $n$, then $G$ has an induced path of order at least $\frac{1}{4} (\log n)^{\frac{1}{3(k-1)(s-1)}}$.
\end{corollary}

Clearly, forbidding $K_{s,s}$-subgraphs (which amounts to forbidding large induced bicliques and large cliques) is a necessary condition in the statement of Corollary~\ref{cor:cliwi}: as we mentioned above, in cliques and bicliques the property \lpp{} only holds with a constant bound.

\paragraph{Well-quasi-ordered induced subgraphs ideals.}
Let $\mathcal{G}$ be a graph class. We say that $\mathcal{G}$ is \emph{hereditary} if any induced subgraph of a graph of $\mathcal{G}$ also belongs to $\mathcal{G}$.
The class $\mathcal{G}$ is \emph{well-quasi-ordered by induced subgraphs} if in every infinite sequence $(G_i)_{i\in\N}$ of graphs from $\mathcal{G}$, there are integers $i<j$ such that $G_i$ is an 
induced subgraph of $G_j$.
Theorem~\ref{th:twintro} also has the following consequence.
\begin{corollary}\label{cor:atminas}
Let $\mathcal{G}$ be a hereditary class of graphs that does neither contain all cliques nor all bicliques. If $\mathcal{G}$ is well-quasi-ordered by induced subgraphs then there is a constant $d \in (0,1)$ such that the following holds: every $G\in \mathcal{G}$ that has a path of order $n$ also has an induced path of order at least $n^d$.
\end{corollary}

This holds because, as proved in \cite[Theorem~1]{atminas2019graphs}, graph classes such as $\mathcal{G}$ in the above statement have pathwidth bounded from above by a constant (that depends on the order of the smallest excluded clique and biclique), so we can apply Theorem~\ref{th:twintro}.\footnote{The bound on the pathwidth appeared in the published version of \cite{atminas2019graphs}. As of the time of writing, the preprint of the paper on arxiv only proves a bound on the treewidth.} Similarly as above the requirement in Corollary~\ref{cor:atminas} that $\mathcal{G}$ does not contain all cliques or bicliques is necessary for hereditary classes. Note nevertheless that the assumption that $\mathcal{G}$ is well-quasi-ordered is only sufficient. For instance, the class of graphs of maximum degree at most 2 (which are unions of cycles and paths) is hereditary, excludes large cliques and bicliques, and satisfies \lpp{} with a linear bound, however it is not well-quasi-ordered by induced subgraphs due to the presence of arbitrarily long cycles.

\paragraph{Graphs with polynomial-sized modulators.}
As we will prove in Section~\ref{sec:bdd}, graphs that are not \emph{too far} from a class where \lpp{} holds with a polylogarithmic bound also satisfy this property with a similar bound. More formally, we will prove the following statement.

\begin{lemma}[Restatement of Lemma~\ref{lem:del}]
Let $\mathcal{G}$ be a hereditary class of graphs where~\lpp{} holds with function $n \mapsto c (\log n)^d$.
Let $\varepsilon \in (0,1)$ and let $\mathcal{G}_\varepsilon$ denote the class of graphs $G$ that have a set of vertices $X$ with $|X| \leq |G|^\varepsilon$ such that $G-X \in \mathcal{G}$.

Then $\mathcal{G}_\varepsilon$ satisfies \lpp{} with function $n \mapsto c'(\log n)^d$, for some constant $c'$ that depends on $c$, $d$, and $\varepsilon$ only.
\end{lemma}

Therefore all our results can be extended to classes that are, informally, at \emph{polynomial distance} from the considered classes.
As an example, the lemma above and Theorem~\ref{th:twintro} together imply the following result.

\begin{corollary}
There is a constant $c \in (0,1)$ such that the following holds.
Let $k \in \N$ and let $\mathcal{G}_k$ denote the class of graphs $G$ that have a set $X \subseteq V(G)$ with $|X| \leq \sqrt{|G|}$ such that $\tw(G-X) \leq k$.
Then for every $G \in \mathcal{G}_k$, if $G$ has a path of order $n$ then it has an induced path of order at least $c (\log n)^{1/k}$.
\end{corollary}
Observe that $\mathcal{G}_k$ above is not closed under topological minors, so it does not fall under the umbrella of Theorem~\ref{th:minclo}. Also, a direct application of Theorem~\ref{th:twintro} to graphs of $\mathcal{G}_k$ (with bound $k + \sqrt{|G|}$ on the treewidth) would not guarantee more than constant-sized induced paths.

\paragraph{Graphs of moderate degree.}
In order to prove Theorem~\ref{th:minclo} we will show that graphs of degree bounded by a constant satisfy property \lpp{} with a polylogarithmic bound (Corollary~\ref{cor:bddeg}). More generally, we can prove that the same holds for graphs with subpolynomial maximum degree (Lemma~\ref{lem:bddelta}). For instance, an application of this lemma shows that the class of the graphs $G$ where the degree is at most $2^{\sqrt{\log |G|}}$ satisfies \lpp{} with bound $n \mapsto \sqrt{\log n}$. Note that this graph class is not closed under taking topological minors.

\subsection{Organization of the paper.}
In Section~\ref{sec:prelim} we give the necessary definitions.
We prove Theorem~\ref{th:pwintro} in Section~\ref{sec:pw} and Theorem~\ref{th:twintro} in Section~\ref{sec:tw}. Section~\ref{sec:minclo} is split into several subsections in order to finally prove Theorem~\ref{th:minclo}. We discuss further research directions in Section~\ref{sec:op}.

\section{Preliminaries}\label{sec:prelim}

Unless otherwise specified, logarithms are binary.
Graphs in this paper are finite, undirected, and loopless.
For every graph $G$ we respectively denote by $V(G)$ and $E(G)$ its sets of vertices and edges and use $|G|$ as a shorthand for $|V(G)|$, the \emph{order} of the graph. We use standard graph theory terminology.

\begin{remark}\label{rem:hairsplit}
If $G$ is a graph that has a path $P$ of order $n$ and $X \subseteq(G)$ is not empty, then $G-X$ has
\begin{enumerate}
\item at most $|X| + 1$ connected components that contain a vertex of $P$; and
\item a connected component that has a path of order at least
$\frac{n-|X|}{|X| + 1} \geq \frac{n}{2|X|} - 1$.
\end{enumerate}
\end{remark}

\paragraph{Long paths versus Hamiltonian paths.}
A \emph{Hamiltonian path} in a graph $G$ is a path that visits all the vertices of~$G$.
The statements of the results in Section~\ref{sec:intro} follow the general form of \lpp{}.
In contrast, in the rest of the paper we work with statements of the following form (for $\mathcal{G}$ a class of graphs):
\begin{enumerate}
\item[\hypertarget{it:hpp}{$(\star\star)$}] \textit{For every $n\in \N$, if a graph of $\mathcal{G}$ has a Hamiltonian path and order $n$ then it has an induced path of order at least $f(n)$.}
\end{enumerate}
\begin{remark}\label{rem:hpp}
For a hereditary graph class $\mathcal{G}$, the statements \lpp{} and \hpp{} are equivalent.
Given a graph $G \in \mathcal{G}$ with a path $P$ of order $n$ (as required by \lpp{}), it suffices to consider the induced subgraph $G[V(P)] \in \mathcal{G}$ to be able to apply statement \hpp{}. This shows that \hpp{} implies \lpp{} and the other direction is trivial.
\end{remark}
However the form \hpp{} is more convenient for the proofs because the induced paths that we construct will never use vertices other than those of the path $P$ whose existence is assumed, so by using form \hpp{} we do not need to explicitly say that we restrict our attention to $G[V(P)]$.

\paragraph{Representations.}

To easily deal with graphs of bounded pathwidth or treewidth (defined hereafter) we find it convenient to define tree representations, which are objects that are closely related to tree decompositions, as we explain below.
Formally, a \emph{tree representation} of a graph $G$ is a pair $\T = (T, \{T_v\}_{v \in V(G)})$ such that
\begin{enumerate}
\item $T$ is a tree;
\item for every $v \in V(G)$, $T_v$ is a subtree of $T$, called the \emph{model} of $v$;\label{it:sbt}
\item for every edge $uv$ of $G$ the subtrees $T_u$ and $T_v$ intersect. \label{it:inter}
\end{enumerate}

When $T$ is a path, we call $\T$ a \emph{path representation}.
If item \eqref{it:inter} is strengthened as follows
\begin{enumerate}
\item[3'.] $uv$ is an edge of $G$ if and only if $T_u$ and $T_v$ intersect 
\end{enumerate}

\noindent then $G$ is the intersection graph of the vertex sets of the subtrees $\{T_v\}_{v \in V(G)}$ and it is a \emph{chordal graph}\footnote{Chordal graphs are usually defined as graphs with no induced cycle of order 4 or more. They can be seen as intersection graphs of subtrees of a tree, as proved by Gavril~\cite[Theorem~3]{GAVRIL197447}.}; if furthermore $T$ is a path then $G$ is an \emph{interval graph}. In this latter case we call $\T$ an \emph{interval representation} of~$G$ to stress that it is an interval graph.
To avoid confusion between the vertices of $G$ and $T$, we use the synonym \emph{nodes} to refer to vertices of~$T$.

Tree representations and tree decompositions are closely linked, as we explain now.
For a tree representation $\T = (T, \{T_v\}_{v \in V(G)})$ of a graph $G$, we define for every $t \in V(T)$ the \emph{bag at $t$} as the following subset of $V(G)$
\[
\beta_{\T} (t) = \{v \in V(G) \mid t \in T_v\}.
\]

We drop the subscript when there is no ambiguity.
The \emph{width} of a tree representation~$\T$ is $\max_{t \in V(T)} |\beta_\T(t)| - 1$.
The \emph{treewidth} of $G$ is the minimum width of any of its tree representations and the \emph{pathwidth} of $G$ is the minimum width of any of its path representations.  It can easily be seen that these definitions coincide with the usual definitions for treewidth and pathwidth as, with the notation above, $(T, \{\beta(t)\}_{t \in V(T)})$ is a tree decomposition of~$G$. We respectively denote by $\tw(G)$ and $\pw(G)$ the treewidth and pathwidth of $G$.

The following is a consequence of items \eqref{it:sbt} and \eqref{it:inter} of the definition of a tree representation.
\begin{remark}\label{rem:conntree}
Let $G$ be a graph with a tree representation $(T, \{T_v\}_{v\in V(G)})$ and let $H$ be a connected subgraph of $G$. Then $\bigcup_{v\in V(H)} T_v$ is a (connected) subtree of~$T$.
\end{remark}

\section{Induced paths in graphs of bounded pathwidth}
\label{sec:pw}

In this section we show that the maximum function $f_k$ such that property \lpp{} holds for graphs of pathwidth less than $k$ is such that $\frac{1}{3}n^{1/k} \leq f_k(n) \leq n^{2/k} + 1$ (Theorem~\ref{th:pwintro} and Theorem~\ref{th:pwlb}).

A \emph{caterpillar} is a tree in which all the vertices are at distance at most one of some path.
\begin{lemma}[{See \cite[Section~6]{dmtcs:263}}]\label{lem:cat}
Every connected graph of pathwidth at most one is a caterpillar.
\end{lemma}

We will also use the following consequence of the Helly property of intervals (see for instance  \cite[Section 2.5]{graham1995handbook}).
\begin{lemma}\label{lem:helly}
Suppose $\mathcal{P} = (P, \{P_v\}_{v \in V(G)})$ is a path representation of a graph $G$ and $K$ is a clique of $G$. Then there is a node $t \in V(P)$ such that $V(K) \subseteq \beta_{\mathcal{P}}(t)$.
\end{lemma}

We first prove Theorem~\ref{th:pwintro}, that we restate here for convenience.
\pwintro*

\begin{proof}
We actually prove the following by induction on $n$ and $k$, which is equivalent to the desired statement according to Remark~\ref{rem:hpp}.
\begin{itemize}
\item[] \textit{If $G$ is a graph of order at least $n$ and pathwidth less than $k$ that has a Hamiltonian path, then $G$ has an induced path of order at least $\frac{1}{3}n^\frac{1}{k}$.}
\end{itemize}
Note that the statement is vacuously true when $k=1$ and $n > 1$ as there is no connected graph with $|G|\geq 2$ and $\pw(G) = 0$.
When $\pw(G)=n-1$ (thus $k = n$), the graph is a clique and the order of its longest induced path is $2 \geq \frac{1}{3}n^\frac{1}{k}$.
In the cases where $n \geq 2$ and $k = 2$, G is a caterpillar, by Lemma~\ref{lem:cat}. Every caterpillar with a Hamiltonian path is a path, so the statement holds as $G$ is already an induced path and $n \geq \frac{1}{3}\sqrt{n}$.

So we now assume that $k \geq 3$, $n > k$ and that the statement holds for all smaller values of $k$ and $n$. We now show that it also true for $k$ and $n$. We may also assume that $\frac{1}{3}n^{\frac{1}{k}} > 2$ as otherwise any edge is an induced path of the required order.

Let $G$ be a graph that has a Hamiltonian path and such that $|G|\geq n$ and $\pw(G) < k$. We fix a path representation $\mathcal{R} = (R, \{R_v\}_{v \in V(G)})$ of~$G$ of width less than $k$ and with no empty bag.
Let $u$ be a vertex of $G$ such that $R_u$ contains one endpoint of $R$ and let $v$ be a vertex whose model contains the other endpoint; possibly $u=v$.
Let $Q$ be an induced path in $G$ between $u$ and $v$. Observe that $Q$ may consist of a single vertex (when $u = v$) or of two vertices (when $u$ and $v$ are adjacent). Otherwise, $Q$ can be constructed by taking a shortest path between $u$ and~$v$.

In the case where $|Q|\geq \frac{1}{3}n^\frac{1}{k}$ we are done and $Q$ is the desired path.
So in the rest of the proof we may assume that $|Q| < \frac{1}{3}n^\frac{1}{k}$.
Let $P$ be a Hamiltonian path of $G$.
The path $Q$ intersects $P$ in $|Q|$ vertices thus removing $Q$ from $G$ cuts $P$ into at most $|Q|+1$ subpaths, and the longest of them, that we call $P'$, has an order $n'$ that is at least $\frac{n-|Q|}{|Q|+1}$. Let us consider $G'$ the graph induced by~$P'$ in $G$.

We now show that $\pw(G') \leq k-1$.
For this we consider the path representation $\mathcal{R}' = (R, \{R_w\}_{w \in V(G')})$ of~$G'$.
Let $r \in R$ be a node such that $|\beta_{\mathcal{R}'}(r)|$ is maximum.
From the definition we have $\beta_{\mathcal{R}'}(r) \subseteq \beta_{\mathcal{R}}(r)$.
By Remark~\ref{rem:conntree}, the definition of $Q$ and the fact that it is connected, we know that the union of the sets $\{R_w\}_{w \in V(Q)}$ is equal to $V(R)$. Therefore $\beta_{\mathcal{R}}(r)$ contains a vertex of~$Q$. This implies $| \beta_{\mathcal{R}}(r) | > | \beta_{\mathcal{R}'}(r) |$ and the claimed bound on the pathwidth of $G'$ follows.

The graph $G'$ has a Hamiltonian path (by definition), order $n'<n$ and pathwidth at most~$k-1$. By induction, $G'$ admits an induced path $Q'$ of order at least $\frac{1}{3} {n'}^\frac{1}{k-1}$.
Recall that $n' \geq \frac{n-|Q|}{|Q|+1}$. Since $|Q| < \frac{1}{3}n^\frac{1}{k}$, we have
$n' > \frac{n}{\frac{1}{3}n^\frac{1}{k} +1}-1$.
As we assume $\frac{1}{6}n^{\frac{1}{k}} \geq 1$ we
deduce $n' > 2 n^\frac{k-1}{k} - 1 \geq n^\frac{k-1}{k}$.
Therefore we have $|Q'| \geq \frac{1}{3} (n^\frac{k-1}{k})^\frac{1}{k-1} = \frac{1}{3}n^\frac{1}{k}$. Since $G'$ is an induced subgraph of $G$, $Q'$ is also an induced path of $G$ so we are done.
\end{proof}

As every interval graph of clique number $k$ admits a path representation of width less than~$k$ (given by its interval representation), we have the following improvement of the bound of Theorem~\ref{th:esp}.

\begin{corollary}\label{cor:intv}
For every $k, n \in \N$, if $G$ is an interval graph of order at least $n$ and clique number at most $k$ that has a Hamiltonian path, then $G$ has an induced path of order at least $\frac{1}{3}n^\frac{1}{k}$.
\end{corollary}

The following statement complements Theorem~\ref{th:pwintro} and Corollary~\ref{cor:intv} by giving an upper-bound on the order of induced paths one can guarantee in (interval) graphs of bounded pathwidth.

\begin{theorem}\label{th:pwlb}
For every $k, n \in \N$ with $2 \leq k \leq n$, there exists an interval graph $G_{n,k}$ with a Hamiltonian path of order at least $n$ and clique number at most $k$, such that every induced path of $G_{n,k}$ has order at most $n^\frac{2}{k} + 1$.
\end{theorem}

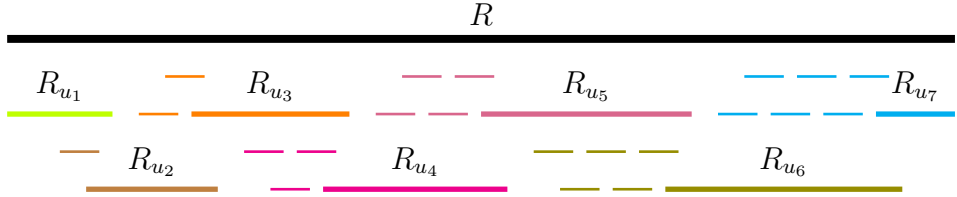
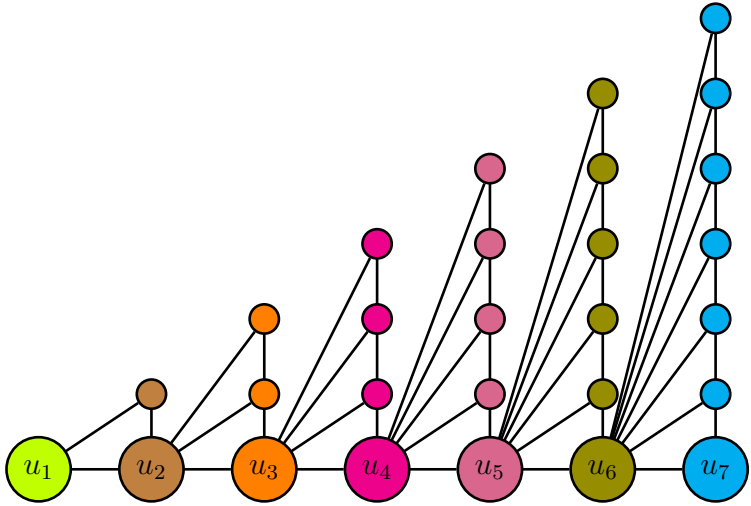
\begin{figure}
    \begin{subfigure}[b]{\textwidth}
    \centering
        \begin{tikzpicture}[xscale = 0.7, every path/.style={line width=1pt}]
            \draw[line width=3pt] (0,2) -- (36/2,2) node[midway,above] {$R$};
             \pgfmathsetmacro{\x}{0}
             \foreach \x/\color [count=\i from 1] in {0/lime,3/brown,7/orange,12/magenta,18/purple!60,25/olive,33/cyan}{
                 \pgfmathsetmacro{\xx}{\x+\i+3}
                 \ifthenelse{\i = 7}{\pgfmathsetmacro{\xx}{\x+3}}{}
                 \draw[draw=\color, line width=2pt] (\x/2,{Mod(\i,2)}) -- (\xx/2,{Mod(\i,2)}) node[midway,above] {$R_{u_\i}$};
                 \ifthenelse{\i>1}{
                     \foreach \ii in {2,...,\i}{
                         \draw[draw=\color] ( \x/2-\ii/2+1/2 , {Mod(\i,2) + Mod(\ii+1,2)/2} ) -- (\x/2-\ii/2+1.25 , {Mod(\i,2) + Mod(\ii+1,2)/2} );
                     }
                 }
            }
        \end{tikzpicture}
        \caption{The interval representation of the graph $G_{n,3}$ with $q=7$.}
        \label{fig:rpzGn3}
     \end{subfigure}
     \medskip{}
     
    \begin{subfigure}[b]{\textwidth}
\centering 
        \begin{tikzpicture}[xscale=1.5,
                            every node/.style = {draw,fill=white,circle},
                            every path/.style={line width=1pt}]
             \foreach \i/\color in {7/cyan,6/olive,5/purple!60,4/magenta,3/orange,2/brown}{
             \draw (\i, 1) node[fill=\color] (u{\i,1}) {$u_\i$}
             \foreach \j in {2, ..., \i}{
                 -- (\i, \j) node[fill=\color] (u{\i,\j}) {}
             };
             \pgfmathsetmacro{\ii}{\i-1}
             \foreach \j in {1, ..., \i}{
             \draw (u{\i,\j}) to (\ii, 1);
             }
             }
            \draw (1,1) node[fill=lime] (u{1,1}) {$u_1$};
        \end{tikzpicture}
        \caption{The graph $G_{n,3}$ with $q=7$.}
        \label{fig:Gn3}
    \end{subfigure}
    \caption{The construction of Theorem~\ref{th:pwlb}.}
\end{figure}

\begin{proof}
The proof is by induction on $n$ and $k$.
When $k = n$, $G_{n , n}$ is the clique on $n$ vertices, where every induced path has order at most 2, which is less than $n^{\frac{2}{n}} + 1$. 
When $k = 2$, $G_{n,2}$ is the (induced) path on $n$ vertices, which again satisfies the desired statement for every $n \geq 2$.

For any $n\geq k$, let $q = \floor{n^{\frac{2}{3}}+1}$. When $k=3$, $G_{n,3}$ is constructed from a collection $\{P_i\}_{i \in \intv{1}{q}}$ of $q$ paths, where for every $i \in \intv{1}{q}$ the path $P_i$ has order $i$ and an endpoint called $u_i$, by connecting $u_i$ to all the vertices of $P_{i+1}$ for every $i \in \intv{1}{q-1}$.
See Figure~\ref{fig:Gn3} for a depiction of a small case and Figure~\ref{fig:rpzGn3} for an interval representation $\mathcal{R} = (R, \{R_v\}_{v \in V(G)})$ of it, where each $R_v$ is drawn below $R$ respecting the x-axis.

We can see that at most 3 intervals intersect ($R_{u_i}$ and the intervals representing two consecutive vertices of $P_{i+1}$ for each $i$), thus $G_{n,3}$ is an interval graph with clique number $3$.

The number of vertices of $G_{n,3}$ is $\sum\limits_{i=1}^{q} |P_i|=\frac{q(q+1)}{2} \ge \frac{n^{\frac{2}{3}}(n^{\frac{2}{3}}+1)}{2}$ since $q\ge n^\frac{2}{3}$, which is greater than~$n$. $G_{n,3}$ admits a Hamiltonian path starting in $u_q$ that, for each~$i$ from~$q$ to~$2$, follows $P_i$ from $u_i$ to its other endpoint, goes to $u_{i-1}$, and repeats the same process.

Let us now bound the maximum order of an induced path in $G_{n,3}$.
Let $Q$ be an induced path of $G_{n,3}$ and let $i$ denote the minimum integer such that $Q$ has a vertex from $P_i$; clearly $Q$ has at most $i$ vertices from this path.
Let $j \in  \intv{i+1}{q}$ and observe that every vertex of $P_j$ has $u_{j-1}$ as unique neighbor in $\{P_{j'}\}_{j'<j}$. As $Q$ is induced we deduce that it contains at most one vertex of $P_j$.
This holds for each of the $q-i$ paths of $\{P_j\}_{j \in \intv{i+1}{q}}$, so we get the bound $|Q| \leq q$, as desired. This concludes the proof for the case $k=3$.

So we now take $n > k \geq 4$, and assume that $G_{n', k'}$ is defined and satisfies the statement for every $k' < k$ and $n' < n$ such that $2\leq k' \leq n'$.
To construct $G_{n,k}$, we proceed as follows.

Let $q = \floor{n^{\frac{2}{k}}+1}$.
If $q \geq n$ then the graph $G_{n,k} = P_q$ clearly satisfies the desired statement.
Otherwise, we set $n' = \ceil{\frac{n-q}{q - 1}}$; observe that $n' \geq 2$.
We construct $G_{n, k}$ from the disjoint union of a path $Q = v_1 \dots v_{q}$ and $q-1$ copies $H_1, \dots, H_{q-1}$ of $G_{n', k-2}$ by connecting $v_{i}$ and $v_{i+1}$ to all vertices of $H_i$, for every $i \in \intv{1}{q-1}$.

To show that $G_{n ,k}$ is an interval graph, we now provide an interval representation of it (see Figure~\ref{fig:kgeq4} for an illustration).
By our induction hypothesis, $H_i$ admits an interval representation $\mathcal{M}^i = (M^i, \{M^i_v\}_{v \in V(H_i)})$ for every $i \in \intv{1}{q-1}$. Let $M$ denote the concatenation of the paths $M^1, \dots, M^{q-1}$. That is, $M$ is obtained from the disjoint union of these paths by adding an edge between an endpoint of $M^1$ and one endpoint of $M^2$, from the other endpoint of $M^2$ to one endpoint of $M^3$, and so on. Clearly $\left (M , \bigcup_{i=1}^{q-1} \{M^i_v\}_{v\in V(H_i)} \right)$ is an interval representation of the disjoint union of the $H_i$'s.
For every $i \in \intv{2}{q-1}$, let $M_{v_i} = M[V(M^{i-1}) \cup V(M^{i})]$ and let $M_{v_1} = M^1$ and $M_{v_{q}} = M^{q-1}$. Observe that $(M, \{M_{v_i}\}_{i \in \intv{1}{q}})$ is an interval representation of $Q$. Furthermore, $\mathcal{M} = (M, \{M_{u_i}\}_{i \in \intv{1}{q}} \cup \bigcup_{i=1}^{q-1} \{M^i_v\}_{v\in V(H_i)})$ is an interval representation of $G_{n, k}$.
This proves that $G_{n,k}$ is an interval graph.

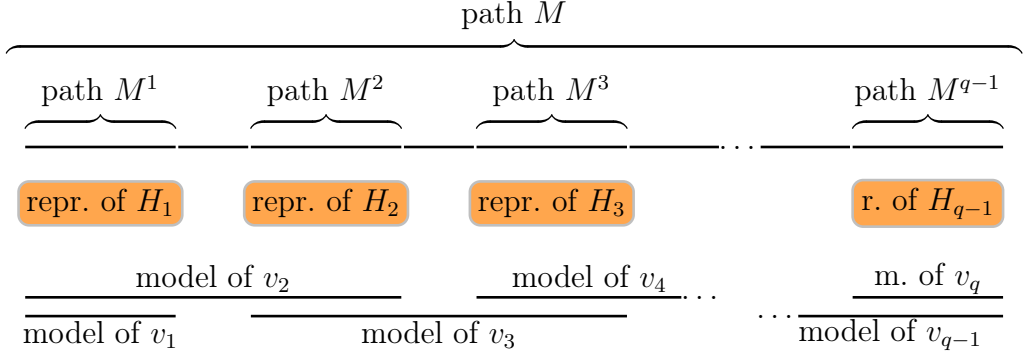
\begin{figure}
\centering
\begin{tikzpicture}[every path/.style = {line width = 1pt},
                    every node/.style = {draw=none, fill=none, inner sep=0}]
\begin{scope}
\draw (0,0) node (f0) {} -- ++(2,0) node (l0) {};
\draw [decorate, decoration = {calligraphic brace, amplitude=5pt}] ([shift={(0,0.25)}]f0.center) -- node[midway, yshift=0.5cm] {path $M^1$} ([shift=({0, 0.25})]l0.center);
\draw (1, -0.75) node[rectangle, rounded corners, draw=lightgray, fill=orange!70, minimum width=2cm, inner sep=3pt] {repr.\ of $H_1$};
\end{scope}

\begin{scope}[xshift=3cm]
\draw(0,0) node (f1) {} -- ++(2,0) node (l1) {};
\draw [decorate, decoration = {calligraphic brace, amplitude=5pt}] ([shift={(0, 0.25)}]f1.center) -- node[midway, yshift=0.5cm] {path $M^2$} ([shift=({0, 0.25})]l1.center);
\draw (1, -0.75) node[rectangle, rounded corners, draw=lightgray, fill=orange!70, minimum width=2cm, inner sep=3pt] {repr.\ of $H_2$};
\end{scope}

\begin{scope}[xshift = 6cm]
\draw (0,0) node (f2) {} -- ++(2,0) node (l2) {};
\draw [decorate, decoration = {calligraphic brace, amplitude=5pt}] ([shift={(0, 0.25)}]f2.center) -- node[midway, yshift=0.5cm] {path $M^3$} ([shift=({0, 0.25})]l2.center);
\draw (1, -0.75) node[rectangle, rounded corners, draw=lightgray, fill=orange!70, minimum width=2cm, inner sep=3pt] {repr.\ of $H_3$};
\end{scope}

\begin{scope}[xshift = 11cm]
\draw (0,0) node (f3) {} -- ++(2,0) node (l3) {};
\draw [decorate, decoration = {calligraphic brace, amplitude=5pt}] ([shift={(0, 0.25)}]f3.center) -- node[midway, yshift=0.5cm] {path $M^{q-1}$} ([shift=({0, 0.25})]l3.center);
\draw (1, -0.75) node[rectangle, rounded corners, draw=lightgray, fill=orange!70, minimum width=2cm, inner sep=3pt] {r.\ of $H_{q-1}$};
\end{scope}

\draw (l0) -- (f1) (l1) -- (f2) (l2) -- node[midway, circle, draw=none, fill=white] {$\dots$} (f3);
\draw [decorate, decoration = {calligraphic brace, amplitude=5pt}] ([shift={(-0.25,1.25)}]f0.center) -- node[midway, yshift=0.5cm] {path $M$} ([shift={(0.25, 1.25)}]l3.center);
\begin{scope}[yshift = -3cm]
\draw ([yshift = -2.25cm]f0.center) -- node[midway, yshift=-0.25cm] {model of $v_1$} ([yshift = -2.25cm]l0.center);
\draw ([yshift = -2cm]f0.center) -- node[midway, yshift=0.25cm] {model of $v_2$} ([yshift = -2cm]l1.center);
\draw ([yshift = -2.25cm]f1.center) -- node[midway, yshift=-0.25cm] {model of $v_3$} ([yshift = -2.25cm]l2.center);
\draw ([yshift = -2cm]f2.center) -- node[midway, yshift=0.25cm] {model of $v_4$} ++(3,0) node[draw=none, circle, fill=white] {$\dots$};
\draw ([yshift = -2.25cm]l3.center) -- node[midway, yshift=-0.25cm] {model of $v_{q-1}$} ++(-3,0) node[draw=none, circle, fill=white] {$\dots$};
\draw ([yshift = -2cm]l3.center) -- node[midway, yshift=0.25cm] {m.\ of $v_{q}$} ([yshift = -2cm]f3.center);
\end{scope}
\end{tikzpicture}
\caption{The construction of an interval representation for $G_{n,k}$ when $k > 3$.}
\label{fig:kgeq4}
\end{figure}

The graph $G_{n,k}$ is composed of the path $Q$ on $q$ vertices and $q-1$ copies of $G_{n',k-2}$ which have (by induction) at least $n'$ vertices each, so we have
\begin{align*}
|G_{n, k}| &\geq q + (q-1) n'\\
    &= q + (q-1) \ceil{\frac{n-q}{q - 1}}\\
    &\geq n.
\end{align*}

We now show that the clique number of $G_{n, k}$ is at most $k$.
Let $K$ be a maximum clique of $G_{n,k}$. By Lemma~\ref{lem:helly}, there exists a node $x$ of $M$ such that $K$ consists of all vertices of $G$ whose model (in $\mathcal{M}$) contains $x$.
Let $j$ be such that $x \in M^j$ (such a vertex exists by definition of $M$), we then have
\begin{align*}
V(K) = &\quad (\; V(K) \cap V(Q)\; ) \quad \cup \quad (\; V(K) \cap V(H_j)\; ).
\end{align*}

The first intersection has size at most 2 since $Q$ is an induced path.
Recall that $H_j$ is a copy of $G_{n',k'}$ thus its clique number is at most $k-2$. Therefore the second intersection has size at most $k-2$. We deduce $|K| \leq (k-2) + 2 = k$.

Let us show that $G_{n, k}$ has a Hamiltonian path.
For every $i$, let $R_i$ denote a Hamiltonian path of $H_i$, which exists by induction hypothesis.
By construction, every vertex of $R_i$ (in particular its endpoints) is adjacent to both $v_{i-1}$ and $v_i$ in $G_{n, k}$.
Therefore $v_1 R_1 v_1 R_2 \dots v_{q-1} R_{q-1} v_{q}$ is a (Hamiltonian) path in $G$.

We proved that $G_{n, k}$ is an interval graph on at least $n$ vertices, with clique number at most $k$, and with a Hamiltonian path. In order to conclude the proof, it remains to prove that $G$ does not have an induced path longer than $n^{2/k} + 1$.
Let $P$ be an induced path of $G$ of maximum length.

We first consider the case where $|V(P) \cap V(H_i)| \geq 2$ for some $i$.
As for  every $w\in V(H_i)$, $N(w) \setminus V(H_i) = \{v_{i-1}, v_i\}$, we deduce that $V(P) \subseteq V(H_i)$, otherwise $P$ would not be induced. By induction hypothesis we get $|P| \leq n'^{\frac{2}{k-2}} + 1$.
Observe that
\begin{align*}
n' & \leq \frac{n-q}{q-1} + 1 &\text{(from the definition)}\\
& \leq \frac{n}{q-1}\\
& \leq n^{\frac{k-2}{k}}.
\end{align*}
So $|P| \leq n^{\frac{2}{k}} + 1$, as required.
We now consider the remaining case where $|V(P) \cap V(H_i)| \leq 1 $ for all $i$.
Notice that an internal vertex $w$ of $P$ cannot belong to $H_i$ for some $i$. Indeed the only neighbors of $w$ outside $H_i$ are $v_{i-1}$ and $v_{i}$, which are adjacent. We then get the three following cases:

\begin{itemize}
\item either $V(P)$ does not intersect any $H_i$, in which case $P = Q$ so $|P| \leq n^\frac{2}{k} + 1$;
\item or there is an $i$ such that $P$ starts at some vertex $w \in V(H_i)$ and does not intersect $H_j$ for all $j \neq i$, in which case $P = w v_i v_{i+1} \dots v_q$ or $P = w v_{i-1} v_{i-2} \dots v_1$ and $|P| \leq |Q| \leq n^\frac{2}{k} + 1$;
\item or there are two integers $i,j$ with $i<j$ such that $P$ starts from some vertex $w_i \in V(H_i)$ and ends at some vertex $w_j$ of $V(H_j)$, in which case $P = w_i v_{i} v_{i+1} \dots v_{j-1} w_j$ and again $|P| \leq |Q| \leq n^\frac{2}{k} + 1$.
\end{itemize}
\end{proof}

\section{Induced paths in graphs of bounded treewidth}
\label{sec:tw}

In this section we show that the maximum function $f_k$ such that property \lpp{} holds for graphs of treewidth less than $k$ is such that $f_k(n) \geq \frac{1}{4}n^{1/k}$ (Theorem~\ref{th:twintro}).

Let $G$ be a graph and let $\T =(T, \{T_v\}_{v\in V(G)})$ be a tree representation of~$G$.
The \emph{weight of a path} $P$ of $T$ is the number of vertices $v$ of $G$ such that $T_v$ intersects $P$. The \emph{weight of a node} $x$ of $T$ is the maximum weight of a path from $x$ to a leaf minus $|\beta_\T(x)|$, and is noted $w_{\T}(x)$ (or $w(x)$ when there is no ambiguity on the tree representation).

It is well-known that the order of a tree is upper-bounded by a function of its height and maximum degree. The following lemma extends this statement to graphs of bounded treewidth that have a Hamiltonian path.

\begin{lemma}\label{lem:weight}
Let $k, w \in \N$ and let $G$ be a graph that has a Hamiltonian path.
If there is a tree representation of $G$ of width less than $k$ that has a node of weight at most $w$, then $|G|\leq (k+1)^{w+1}-1$.
\end{lemma}

\begin{proof}
Let $k \in \N$.
The proof is by induction on $w$.

If $w = 0$ then for every graph $G$ and tree representation $\T$ as in the statement of the lemma and $x$ node of weight zero we have $V(G)\subseteq \beta_\T(x)$ and there are at most $k$ vertices in $\beta_\T(x)$ so the claimed bound holds.
So we now suppose that $w\geq 1$ and that the statement is true for every weight~$w' < w$.

We consider a graph $G$ with a Hamiltonian path $P$ and a tree representation $\T =(T, \{T_v\}_{v \in V(G)})$ of width less than $k$ that has a node $x$ such that $w_\T(x)\leq w$.
Let us consider the graph $G\setminus \beta_\T(x)$. Removing the at most $k$ vertices of $\beta_\T(x)$ cuts $P$ into $t\leq k+1$ subpaths $P_1,\dots P_t$.
For each $i$, let $G_i$ be the graph induced by $P_i$ and let $T_i$ be the union of the $T_v$'s for $v \in V(G_i)$. Then $T_i$ is a subgraph of $T$ and by Remark~\ref{rem:conntree} it is connected. Observe that $\T_i=(T_i, \{T_v\}_{v\in V(G_i)})$ is a tree representation of~$G_i$.

Let $x_i$ be the node of $T_i$ that is the closest to $x$ in $T$.
Let $w_i$ be the weight of $x_i$ in $T_i$ and $Q_i$ a path in $T_i$ of maximum weight from $x_i$ to a leaf $l$. Let us consider the path $Q$ in $T$ from $x$ to $l$; note that $Q_i$ is a subpath of $Q$.
So $x_i$ belongs to $Q$, and by construction of $T_i$, $\beta_{\T_i}(x_i)$ is not empty and belongs to $G\setminus \beta_\T(x)$ thus there is at least one vertex in $\beta_{\T}(x_i)\setminus \beta_\T(x)$.
This implies that $w_\T(x)\geq w_i + 1$, and thus $w_i\leq w_\T(x)-1 \leq w-1$.

By construction, $G_i$ admits a Hamiltonian path and we just proved that in the tree representation $(T_i, \{T_v\}_{v\in V(G_i)})$ of width less than $k$ there is a node $x_i$ of weight at most~$w-1$. By induction, $G_i$ has at most $(k+1)^{w}-1$ vertices.
As $V(G) = \beta_\T(x) \cup \bigcup_{i=1}^{t} V(G_i)$ with $t \leq k+1$ and $|\beta_\T(x)|\leq k$ we get $|G|\leq (k+1)( (k+1)^w - 1 )+ k = (k+1)^{w+1} -1$.
\end{proof}

\begin{corollary}\label{cor:height}
In every tree representation of width less than $k$ of a graph of order $n$ there is a node of weight at least~$\log_{k+1}(n + 1) - 1$.
\end{corollary}

In a graph $G$, the \emph{contraction} of an edge $uv$ is the operation that creates a new vertex $w$ adjacent to the neighbors of $u$ and $v$ and then deletes $u$ and $v$. We say that a graph $H$ is a \emph{contraction} of a graph $G$ if $H$ can be obtained from $G$ after a (possibly empty) sequence of edge contractions.
In the sequel we use Corollary~\ref{cor:height} to extract a graph with a long path and bounded pathwidth from a graph with a large path and bounded treewidth. The obtained graph will be a contraction of the original one, which is interesting for us because of the following property.

\begin{remark}\label{rem:contrind}
Let $H$ be an induced subgraph or a contraction of a graph $G$. If $H$ has an induced path of order $n$, then so does $G$.
\end{remark}


\begin{lemma}\label{lem:minorpww}
Let $k,w \in \N$ and let $G$ be a graph that has a Hamiltonian path and order~$n$.
If $G$ admits a tree representation of width less than $k$ that has a path of weight $w$, then there is a contraction of $G$ that has a Hamiltonian path and  is of order $w$ and pathwidth less than $k$.
\end{lemma}

\begin{proof}
Let $\T = (T, \{T_v\}_{v \in V(G)})$ be the tree representation as in the statement of the lemma and $R$ the path of $T$ of weight $w$. Let $P$ be a Hamiltonian path of~$G$.
We prove the statement by induction on the number $p$ of vertices $v$ of $G$ such that $V(T_v) \cap V(R) = \emptyset$.
In the case $p=0$, for every vertex $v\in V(G)$ the subtree $T_v$ intersects $R$. By the properties of tree representations, we have:
\begin{itemize}
\item for every $v\in V(G)$, $V(T_v)\cap V(R)$ induces a (connected) subpath of $R$, that we call $R_v$; and
\item for every $u,v \in V(G)$, $T_u$ and $T_v$ share a vertex if and only if they share a vertex of $R$.
\end{itemize}
Therefore $(R, \{R_v\}_{v\in V(G)})$ is a path representation of $G$. Clearly it has $w$ vertices and width less than $k$ hence we are done.

So we may assume in the sequel that $p>0$ and that the statement holds for all values smaller than~$p$. As $p > 0$ and $P$ is a Hamiltonian path, there is an edge $uv \in E(P)$ such that $T_u$ intersects $R$ while $T_v$ does not.
Let us consider the graph $G'$ obtained by contracting $uv$ into a new vertex $y$ and setting $T_y = T_u \cup T_v$ (which is connected as $T_u$ and $T_v$ intersect).
We call $\T' = (T, \{T_z\}_{z \in V(G')})$ the corresponding tree representation.
Observe that for every $t \in V(T)$,
we have $\beta_{\T'}(t) = \beta_\T(t)$ if $t \notin T_u \cup T_v$ and $\beta_{\T'}(t) = \beta_\T(t) \setminus \{u,v\} \cup \{y\}$ otherwise. Therefore the width of $\T'$ is less than $k$.
Also, observe that the weight of $R$ in $\T'$ is still $w$. Applying the induction hypothesis on $G'$ yields the desired result.
\end{proof}

We are now ready to prove Theorem~\ref{th:twintro}, that we restate here for convenience.

\twintro*

\begin{proof}
The statement that we actually prove is the following, which implies the desired statement according to Remark~\ref{rem:hpp}.

\begin{itemize}
\item[] \textit{For every $k,n \in \N$, if $G$ is a graph of order at least $n$ and treewidth less than $k$ that has a Hamiltonian path, then $G$ has an induced path of order at least
$\frac{1}{4}(\log n)^{1/k}$.}
\end{itemize}

Let $k$, $n$, and $G$ be as in the statement above. The cases $n \leq 2$ or $k \leq 2$ are trivial and the case $n=k$ is handled by Theorem~\ref{th:pwintro} (as then $G$ has pathwidth at most $k$), so we suppose
$n>k\geq 3$.

By combining Corollary~\ref{cor:height} and Lemma~\ref{lem:minorpww} we obtain a contraction $G'$ of $G$ of order at least $\log_{k+1}(n+1)$ with pathwidth less than $k$ and that has a Hamiltonian path.

By Theorem~\ref{th:pwintro}, $G'$ admits an induced path $Q$ of order 
\[
|Q| \geq \frac{1}{3}(\log_{k+1}(n+1))^\frac{1}{k}
    \geq \frac{1}{3\cdot c} (\log (n+1))^{1/k}
\]
where $c < 1.3$ is the maximum of the function $k \mapsto \log(k+1)^{1/k}$. 
By Remark~\ref{rem:contrind} we deduce that $G$ has an induced path of the same order and we are done.
\end{proof}

\section{Induced paths in topological minor-closed classes}\label{sec:minclo}

In this section we use the decomposition theorem of Grohe and Marx for graphs excluding a topological minor (Theorem~\ref{th:structexcltopomin} in this paper) in order to prove Theorem~\ref{th:minclo}.
According to Grohe and Marx' result, such graphs admit tree decompositions where the bags\footnote{Actually the torsos, to be defined in one of the following sections.} are required to come from some prescribed graph classes.

We first show that \lpp{} holds with a polylogarithmic bound for graphs from these classes (Sections~\ref{sec:bdd} and \ref{sec:vort}) and then that it does too in tree representations where, intuitively, the interaction between bags is low (Section~\ref{sec:bda}).
These results are combined in Section~\ref{sec:sum} to finally prove Theorem~\ref{th:minclo}.

\subsection{Almost bounded degree graphs}\label{sec:bdd}

Let $\Delta, k \in \N$. We say that a graph $G$ has \emph{$(k,\Delta)$-almost bounded degree} if $G$ has a set of at most $k$ vertices whose removal yields a graph of maximum degree at most $\Delta$.
In this section we show that such graphs satisfy property \lpp{} with a logarithmic bound (that depends on $k$ and $\Delta$), that is the following lemma.

\begin{lemma}\label{lem:abdd}
  For every $\Delta, k\in \N$ there is a constant $c\in \R^+$ such that if a graph with $(k, \Delta)$-almost bounded degree has a path of order $n$, then it has an induced path of order at least $c \log n$.
\end{lemma}

This result is a direct consequence of Corollary~\ref{lem:cdel} (for deleting a constant number of vertices) and Corollary~\ref{cor:bddeg} (for graphs of bounded degree) that we prove below. They actually follow from more general statements, that we leave here as they may be useful in order to prove Conjecture~\ref{conj:esp}.

\begin{lemma}\label{lem:del}
Let $\mathcal{G}$ be a hereditary class of graphs such that for some $c,d\in \R^+$,
if a graph $G \in \mathcal{G}$ has a Hamiltonian path and order $n$, then $G$ has an induced path of order at least $c (\log n)^d$.

Let $\varepsilon \in (0,1)$
and let $\mathcal{G}_\varepsilon$ denote the class of graphs such that for every $G \in \mathcal{G}_\varepsilon$ there is a subset $X\subseteq V(G)$ of at most $n^{\varepsilon}$ vertices such that $G - X \in \mathcal{G}$.
Then there is a constant  $c'\in \R^+$ depending on $c$, $d$, and $\varepsilon$ such that
if a graph $G \in \mathcal{G}_\varepsilon$ has a Hamiltonian path and order $n$, then $G$ has an induced path of order at least $c' (\log n)^{d}$.
\end{lemma}

\begin{proof}
Let $G \in \mathcal{G}_\varepsilon$ and let $P$ be a Hamiltonian path of $G$. We may assume $n \geq 3$ otherwise $P$ is already an induced path of the desired length.
Let $X \subseteq V(G)$ be such that $G-X \in \mathcal{G}$ and $1 \leq |X| \leq n^{\varepsilon}$, which exists by definition of~$\mathcal{G}_\varepsilon$.
By Remark~\ref{rem:hairsplit}, $G - X$ has a component $H$ with a path of order at least $n' = \frac{n - |X|}{|X| + 1} \geq \frac{n}{2|X|}-1$. As $\mathcal{G}$ is hereditary, $H \in \mathcal{G}$ so it has an induced path of order at least

\begin{align*}
c (\log n')^d &\geq c \left (\log \left (\frac{n^{1-\varepsilon}}{2} - 1 \right ) \right )^d\\
    & \geq c' (\log n)^d
\end{align*}
for a suitable choice of the constant $c'>0$ (depending on $c$, $d$, and $\varepsilon$).
This path is an induced subgraph of $G$ (by Remark~\ref{rem:contrind}) so we are done.
\end{proof}

\begin{corollary}\label{lem:cdel}
Let $\mathcal{G}$ be a hereditary class of graphs such that for some $c,d\in \R^+$,
if a graph $G \in \mathcal{G}$ has a Hamiltonian path and order $n$, then $G$ has an induced path of order at least $c (\log n)^d$.

Let $k \in \N$
and let $\mathcal{G}_k$ denote the class of graphs such that for every $G \in \mathcal{G}_k$ there is a subset $X\subseteq V(G)$ of order at most $k$ vertices such that $G - X \in \mathcal{G}$.
Then there is a constant  $c'\in \R^+$ depending on $c$, $d$, and $k$ such that
if a graph $G \in \mathcal{G}_k$ has a Hamiltonian path and order $n$, then $G$ has an induced path of order at least $c' (\log n)^{d}$.
\end{corollary}

\begin{lemma}\label{lem:bddelta}
Let $c \in \R^+$, $d \in [0,1)$.
Let $\mathcal{G}$ be the class containing every graph $G$ with maximum degree at most $2^{c (\log |G|)^d}$. 
If a graph $G \in \mathcal{G}$ has a Hamiltonian path and order $n$, then $G$ has an induced path of order at least $\frac{1}{c}(\log n)^{1 - d}$.
\end{lemma}

\begin{proof}
Let $P$ be a Hamiltonian path of $G$.
Let $Q$ be an induced path of $G$ of maximum order and let $u$ be one of its endpoints.
Let $q = |Q|$.
For every $i \in \intv{0}{q-1}$, let $D_i$ denote the set of vertices at distance exactly $i$ from $u$.
As $D_i \subseteq N(D_{i-1})$ for every $i \in \intv{1}{q-1}$ we get $|D_i| \leq |D_{i-1}| \cdot 2^{c (\log n)^{d}}$. Therefore
\begin{align*}
n & = \sum_{i=0}^{q-1} |D_i| & \text{(by maximality of $Q$)}\\
& \leq \sum_{i=0}^{q-1} \left ( 2^{c (\log n)^d} \right )^i\\
& \leq \left ( 2^{c  (\log n)^d} \right )^q\\
\text{so}\quad q &\geq \frac{1}{c}(\log n)^{1 - d}.
\end{align*}
\end{proof}

The following corollary about graph classes of degree bounded by a constant $\Delta$ can be obtained from the previous lemma by taking $c = \log(\Delta)$ and $d = 0$.
\begin{corollary}\label{cor:bddeg}
Let $\Delta \in \N$. Every graph with maximum degree at most $\Delta$ that has a Hamiltonian path and order $n$ has an induced path of order at least $\frac{\log n}{\log \Delta}$.
\end{corollary}

\subsection{Escaping the vortices}\label{sec:vort}

Similar to the concept of tree representation, we can define a \emph{cycle representation} of a graph $G$ as a pair $\mathcal{C} = (C, \{C_{v}\}_{v \in V(G)})$ such that:
\begin{enumerate}
\item $C$ is a cycle;
\item for every $v \in V(G)$, $C_v$ is a connected subgraph of $C$; and 
\item for every edge $uv$ of $G$ the subgraphs $C_u$ and $C_v$ intersect.
\end{enumerate}
The notions of \emph{bag} and \emph{width} of a cycle representation are defined similarly as for tree representations.

Let $G_0$ be a graph embedded in a surface $\Sigma$. Let $C$ be a facial cycle of $G_0$.
A \emph{$C$-vortex} is a cycle representation $(C, \{C_v\}_{v \in V(H)})$ of a graph $H$ such that $V(H) \cap V(G_0) = V(C)$ and $v \in C_v$ for every $v \in V(C)$. Note that $C$ is both a subgraph of $H$ and the graph where the representation of $H$ is defined.

For $g, p, a, k\in \N$, a graph $G$ is \emph{$(g, p, a, k)$-almost-embeddable} if for some set $A \subseteq V(G)$ with $|A| \leq a$ there are graphs $G_0, \dots, G_s$ with $s \leq p$ such that 
\begin{enumerate}
\item $G - A = G_0 \cup G_1 \cup \dots \cup G_s$;
\item $G_1, \dots, G_s$ are vertex-disjoint;
\item $G_0$ can be embedded in a surface of Euler genus at most $g$;
\item there are $s$ pairwise vertex-disjoint facial cycles $F_1, \dots, F_s$ of $G_0$ in this embedding, and 
\item for every $i \in \intv{1}{s}$, $G_i$ has an $F_i$-vortex of width less than $k$.
\end{enumerate}

This notion was introduced for the purpose of the proof of the Graph Minor Structure Theorem of Robertson and Seymour \cite{robertson2003graph} and is also used in the decomposition theorem of Grohe and Marx on which we rely.

For $(g, 0, 0, 0)$-almost embeddable graphs, which by definition are graphs of Euler genus at most~$g$, property \lpp{} is known to hold with a polylogarithmic bound as proved by Esperet et al. (Theorem~\ref{th:espgenus}).
We use this result as a base case to show the following more general statement.

\begin{lemma}\label{lem:almemb}
  For every $g, p, a,k\in \N$ with $k\geq 2$ there is a constant $c$ such that the following holds. If $G$ is a $(g, p, a, k)$-almost-embeddable graph that has a Hamiltonian path and order $n$, then $G$ has an induced path of order at least $c \left (\log n \right )^{\frac{1}{k}}$.
\end{lemma}

\begin{proof}
  We may assume that $n\geq 3$ otherwise the statement is trivial.
  Let us first assume that $a=0$.
  In the case where $p = 0$ then $G$ has Euler genus at most $g$ and the result follows from Theorem~\ref{th:espgenus}.  So we assume that $p \geq 1$.
  Let $G_0, \dots, G_s$ and $F_1, \dots, F_s$ be defined as above.

  Suppose first that for some $i \in \intv{1}{s}$, $|F_i| \geq (k + 1)(\log n + 1)$.
  Let $\mathcal{C} = (F_i, \{C_v\}_{v\in V(G_i)})$ be an $F_i$-vortex of $G_i$ of width less than $k$ (given by the definition of $(g, p, a, k)$-almost-embeddable graphs).
  Let $u \in V(F_i)$.
  The bag $\beta_{\mathcal{C}}(u)$ has size at most $k$ so some connected component of $F_i \setminus  \beta_{\mathcal{C}}(u)$ contains a subpath $F$ of $F_i$ of order at least $\log n$.
  Let $H = G[V(F)]$.
  This graph has $F$ as a Hamiltonian path. Observe that by definition of $F$, no model (in $\mathcal{C}$) of a vertex of $F$ contains $u$. In other words, the models of the vertices of $F$ are all subpaths of $F_i \setminus \{u\}$. 
  Hence $(F_i \setminus \{u\}, \{C_v\}_{v \in V(H)})$ is a path representation of width less than $k$ of $H$.
  By Theorem~\ref{th:pwintro} we get an induced path of order at least $\frac{1}{3} (\log n)^{1/k}$ in $H$ hence in $G$. We are not completely done yet but let us now consider the second case before concluding.

  In the second case we assume that $|F_i| < (k+1)(\log n + 1)$ for every $i \in \intv{1}{s}$.
  This implies that $|G_i| < k(k+1)(\log n + 1)$.
  Let $X = \bigcup_{i=1}^{s}V(G_i)$, then $1 \leq |X| < pk(k+1)(\log n + 1)$.
  
  By Remark~\ref{rem:hairsplit}, $G-X$ has a connected component $G'$ that has a Hamiltonian path and order $n'$ where $n' \geq \frac{n}{2|X|}-1$. So $n' \geq n^\varepsilon$ for some constant $\varepsilon \in (0,1)$ that depends on $p$ and $k$ only.
  
  Because it is a subgraph of $G_0$, the graph $G'$ has Euler genus at most $g$.
  Applying Theorem~\ref{th:espgenus} there is a constant $c_g$ depending on $g$ only such that $G'$ (hence $G$) has an induced path of order $q$ at least
  \begin{align*}
    q & \geq c_g \sqrt{\log n'}\\
      & \geq c_g \sqrt{\varepsilon} \sqrt{\log n}\\
      & \geq c_g \sqrt{\varepsilon} (\log n)^{1/k} & \text{as}\ k \geq 2.
  \end{align*}

Let $c = \min(1/3, c_g\sqrt{\varepsilon})$. In both cases we obtained an induced path of order at least $c (\log n)^{1/k}$, as claimed.

The case where $a \geq 1$ follows from Corollary~\ref{lem:cdel} applied to the case where $a=0$.
This concludes the proof.
\end{proof}

\subsection{Representations of bounded adhesion}\label{sec:bda}

In this section we build upon the ideas developed in Sections \ref{sec:pw} and \ref{sec:tw} to show that, roughly, if $\mathcal{G}$ is a class of graphs where \lpp{} holds with a polylogarithmic bound then the same can be said of graphs obtained by gluing together in a tree-like fashion graphs from $\mathcal{G}$ (Lemma~\ref{lem:trga}). This is a crucial step towards the proof of Theorem~\ref{th:minclo}.

Let $\T = (T, \{T_v\}_{v\in V(G)})$ be a tree representation of a graph $G$.
The \emph{adhesion set} of an edge $tt'$ of $T$ is defined as the following subset of $V(G)$:
\[
\adh_\T(tt') = \{v \in V(G)\mid\ \{t,t'\} \subseteq V(T_v)\}.
\]
Equivalently, it can be defined as the intersection of the bags at $t$ and $t'$, that is $\adh_\T(tt')  = \beta_\T(t) \cap \beta_\T(t')$. We drop the subscript when it is clear from the context.
The \emph{adhesion} of $\mathcal{T}$ is the maximum size of the adhesion set of an edge of~$T$.
%
For every $t \in V(T)$, the \emph{torso} of $\mathcal{T}$ at $t$ is the graph obtained from $G[\beta(t)]$ by adding all edges $uv$ such that $u,v \in \adh_\T(tt')$ for some neighbor $t'$ of~$t$. If for every $t \in V(T)$, the torso of $\T$ at $t$ belongs to some graph class $\mathcal{G}$, we say that $\T$ \emph{has torsos from} $\mathcal{G}$.

Torsos and adhesion sets provide two different ways to restrict tree representations.
Observe that graphs of treewidth (respectively pathwidth) less than $k$ are simply graphs that admit a tree representation (respectively path representation) with torsos from the class of graphs of order at most~$k$. For every graph class $\mathcal{G}$ and integer $a$, we respectively denote by $\PR(\mathcal{G}, a)$ and $\TR(\mathcal{G}, a)$ the class of graphs that admit a path representation or a tree representation with torsos from $\mathcal{G}$ and adhesion less than~$a$. We denote by $\TR(\mathcal{G})$ the class of graphs that admit a tree representation with torsos from $\mathcal{G}$ and no restriction on the adhesion. Notice that if the graphs in $\mathcal{G}$ have cliques of bounded order, the adhesion of tree representations of graphs from $\TR(\mathcal{G})$ is implicitly bounded (i.e. $\TR(\mathcal{G}) \subseteq \TR(\mathcal{G}, a)$ for some $a \in \N$).

The next remark easily follows from the definition of tree representations.
\begin{remark}\label{rem:adhsep}
Let $\T = (T, \{T_v\}_{v\in V(G)})$ be a tree representation of a graph $G$, let $tt'\in E(T)$ and let $F,F'$ be the two connected components of $T\setminus \{tt'\}$.
Let $u$ and $u'$ be such that $V(T_u) \cap V(F) \neq \emptyset$ and $V(T_{u'}) \cap V(F') \neq \emptyset$.
Then either one of $u, u'$ belongs to $\adh_\T(tt')$, or $G\setminus \adh_\T(tt')$ has no path from $u$ to $u'$.
\end{remark}

For the purpose of the proof of Lemma~\ref{lem:trga} we need to relate the order of a graph with the length of a path where it is represented. It is not true in general that the existence of a path representation with a long path implies that the represented graph is large; for instance all the bags in this representation could be identical and small. We show below (Lemma~\ref{lem:longdecbiggraph}) that such a statement holds if we require the  considered path representation to satisfy an extra property, being varied, that we define now.

Let $\mathcal{T} = (T, \{T_v\}_{v\in V(G)})$ be a tree representation of a graph $G$.
For an edge $tt'$ in $T$, the tree representation $(T', \{T'_v\}_{v \in V(G)})$ obtained from $\T$ by \emph{contracting} $tt'$ is defined as follows:
\begin{itemize}
\item $T'$ is obtained from $T$ by contracting $tt'$;
\item for every $v \in V(G)$, $T'_v = T_v$ if $\{t,t'\} \cap T_v = \emptyset$ and otherwise $T'_v = T_v \setminus \{t,t'\} \cup \{t''\}$. 
\end{itemize}
Intuitively, we merge the nodes $t$ and $t'$ both in the tree of the representation and in the models of the vertices.

Let us say that $\mathcal{T}$ is \emph{varied} if no bag is a subset of a neighboring bag, i.e.\ for every $tt' \in E(T)$, $\beta(t) \nsubseteq \beta(t')$. In particular, unless $G$ has no vertex, no bag is empty.

Given a tree representation $\T = (T, \{T_v\}_{v\in V(G)})$, it is possible to produce a varied tree representation by iteratively contracting in $\T$ the edges $tt'$ such that $\beta(t) \subseteq \beta(t')$.
Observe that this process changes neither the width or the adhesion of $\T$ nor the fact that it has bags from some specific class of graphs.

\begin{lemma}\label{lem:longdecbiggraph}
Let $G$ be a graph on at least one vertex. If $\cR =  (R, \{R_v\}_{v\in V(G)})$ is a varied path representation of $G$, then $|G| \geq |R|$.
\end{lemma}
\begin{proof}
We prove the following statement by induction on $\ell$.
\begin{itemize}
\item[] \textit{For every $\ell \in \N^+$, if a graph $G$ on at least one vertex admits a varied path representation on a path of order $\ell$, then $|G| \geq \ell$.}
\end{itemize}
The case $\ell=1$ is trivial as we require that $G$ is not empty.
So let us assume that $\ell > 1$ and that the statement holds for smaller values.
Let $G$ be a graph that admits a varied path representation $\cR = (R, \{R_v\}_{v\in V(G)})$ with $|R| = \ell$, and let $r_1\dots r_\ell$ be the vertices of $R$ in the order of the path.
Let $G^- = G[\beta(r_1) \cup \dots \cup \beta(r_{\ell-1})]$. As $\cR$ is varied, $\beta(r_1) \neq \emptyset$ so $G^-$ has at least one vertex. Observe that it admits a varied path representation with $\ell-1$ nodes (for instance $(R\setminus \{r_\ell\}, \{R_v\}_{v\in V(G^-)})$), so by induction it has at least $\ell-1$ vertices. 
As $\cR$ is varied, there is a vertex $v \in \beta(r_\ell)$ that does not belong to $\beta(r_{\ell-1})$. This vertex does not belong to $G^-$ either (as $R_v$ is connected) so $|G| \geq \ell$, as claimed.
\end{proof}

The above lemma allows us to prove the following variant of Theorem~\ref{th:pwintro} on varied path representations of bounded adhesion.

\begin{lemma}\label{lem:adhenough}
Let $G$ be a graph that has a Hamiltonian path.
If $G$ admits a varied path representation with adhesion less than $a$ and at least $\ell$ nodes, then $G$ has an induced path of order at least $\frac{1}{3}\ell^{\frac{1}{2a}}$.
\end{lemma}

\begin{proof}
Let $\cR = (R, \{R_v\}_{v \in V(G)})$ be a path representation as in the statement.
For every $r \in V(R)$, let $Z_r$ be a subset of $V(G)$ of minimum size such that for every neighbor $r'$ of $r$, $\adh(rr') \subseteq Z_r$ and there is a vertex $v \in Z_r$ that does not belong to $\beta(r')$.
Let us show that this set is well-defined and small.
Either the bag at $r$ has a vertex $v$ that does not appear in the bags of any of its neighbors, in which case $\{v\} \cup \bigcup_{r' \in N(r)} \adh(rr')$ satisfies the above properties, or it does not and then $\beta(r)$ is suitable, as $\mathcal{R}$ is varied. Observe that in this case $\beta(r) = \bigcup_{r' \in N(r)} \adh(rr')$. Recall that in $R$ the vertex $r$ has up to two neighbors. So in both cases we have $|Z_r| \leq 2(a-1) + 1 = 2a-1$.

Let $P$ be a Hamiltonian path of~$G$.
An edge of $P$ is called \emph{superfluous} if it has at least one endpoint outside $\bigcup_{r \in V(R)} Z_r$.

\begin{remark}
Suppose that $xy$ is a superfluous edge and let $G'$ be the graph obtained by contracting $xy$ into a new vertex $z$. Then $\cR' =  (R, \{R_v\}_{v \in V(G')})$ is a path representation of~$G'$ (with $R_z = R_x \cup R_y$).
We show that additionally $\cR'$ is a varied representation. Towards a contradiction, let us assume that there are $r,r' \in V(R)$ such that $\beta_{\cR'}(r) \subseteq \beta_{\cR'}(r')$.
Then for one of $x$ and $y$, say $x$, we have $\beta_{\cR}(r) \setminus \beta_{\cR}(r') = \{x\}$, which implies $y \in \adh_{\cR}(rr')$ (as $xy$ is an edge). We just proved that $x,y \in Z_r$, which is a contradiction with the choice of $xy$.
So $\cR'$ is varied.
\end{remark}

Let $H$ be the graph obtained from $G$ after iteratively contracting all the superfluous edges and let $\cR_H = (R, \{R_v\}_{v\in V(H)})$ be the corresponding varied path representation constructed as in the above remark.
By Lemma~\ref{lem:longdecbiggraph}, $|H| \geq |R| \geq \ell$.
All the edges that were contracted did belong to $P$ so $H$ has a Hamiltonian path.
No edge of $H$ is superfluous (by definition) so for every $r \in V(R)$, $\beta_{\cR_H}(r) = Z_r$.
Therefore $\cR_{H}$ is a path representation of $H$ of width at most $2a-1$.
Applying Theorem~\ref{th:pwintro} to $H$ we get an induced path of order at least $\frac{1}{3}\ell^\frac{1}{2a}$. As $H$ is a contraction of $G$, such a path also exists in~$G$ (Remark~\ref{rem:contrind}) hence we are done.
\end{proof}

A graph class is said to be \emph{closed under taking subgraphs} if every subgraph of a graph of the class also belongs to the class. We are now ready to prove the main result of this section.
\begin{lemma}\label{lem:trga}
  Let $a \in \N^+$, $c \in (0,1/3]$ and $d \in (0,1]$.
Let $\mathcal{G}$ be a class of graphs that is closed under taking subgraphs and such that for every $n \in \N$, every $G\in \mathcal{G}$ that has a Hamiltonian path and order $n$ has an induced path of order at least $c (\log n)^d$.

If a graph $G$ of $\TR(\mathcal{G}, a)$ has a Hamiltonian path and order $n$, then it has an induced path of order at least $c (\log n)^{\frac{1}{4a + \frac{1}{d}}}.$
\end{lemma}

\begin{proof}
Let $\T=(T, \{T_v\}_{v \in V(G)})$ be a varied tree representation of $G$ witnessing that $G \in \TR(\mathcal{G}, a)$. Let us fix $\varepsilon = \frac{4ad}{4ad + 1}<1$. Observe that we may assume
\begin{equation}
c (\log n)^{\frac{1}{4a + \frac{1}{d}}} > 2 \label{eq:tribound}
\end{equation}
as otherwise the statement holds trivially. 
Let $P$ be the Hamiltonian path of $G$.

Let us first suppose that $T$ has a node $t$ such that the bag $\beta(t)$ has order at least $n^{\frac{1}{(\log n)^\varepsilon}}$.
Let $H$ denote the graph obtained from $G$ by iteratively contracting every edge of $P$ that has at least one endpoint outside $\beta(t)$.
By construction, $H$ is a supergraph of $G[\beta(t)]$ of the same order and every edge $uv \in E(H) \setminus E(G[\beta(t)])$ is such that $u,v \in \adh_\T(tt')$ for some neighbor $t'$ of~$t$. Therefore, $H$ is a subgraph of the torso of $\T$ at $t$. As $\mathcal{G}$ is subgraph-closed, we deduce $H \in \mathcal{G}$.
Besides, these contractions were applied to edges of $P$ so they yield a Hamiltonian path in $H$. By the properties of $\mathcal{G}$ we deduce that $H$ has an induced path of order at least
\begin{align*}
c \left (\log n^{\frac{1}{(\log n)^\varepsilon}} \right )^{d}%
&= c (\log n)^{(1-\varepsilon)d}\\
&= c (\log n)^{\frac{1}{4a + \frac{1}{d}}}.
\end{align*}
Such a path also exists in $G$ (Remark~\ref{rem:contrind}) so we are done.

So we may assume now that at every node $t$ of $T$ we have
\begin{equation}
|\beta(t)| <  n^{\frac{1}{(\log n)^\varepsilon}}, \label{eq:bagbound}
\end{equation}
which implies
\begin{equation}
|T| \geq \frac{n}{n^{\frac{1}{(\log n)^\varepsilon}}} = n^{1- {\frac{1}{(\log n)^\varepsilon}}}. \label{eq:tbound}
\end{equation}

Let $t \in V(T)$.
As $\T$ is varied, for every neighbor $t'$ of $t$ there is a vertex $v_{t'} \in \beta(t')$ of $G$ such that $v_{t'} \notin \beta(t)$. Let $G_{t'}$ denote the connected component of $G\setminus \beta(t)$ that contains $v_{t'}$. Then for every neighbor $t'' \neq t'$ of $t$, the components $G_{t'}$ and $G_{t''}$ are distinct, by the properties of tree representations. So $G \setminus \beta(t)$ has at least $\deg_T(t)$ connected components.
By \eqref{eq:bagbound} and Remark~\ref{rem:hairsplit}, $G \setminus \beta(t)$ has at most $|\beta(t)|+1$ connected components.
We deduce that the maximum degree $\Delta$ of $T$ is bounded as follows
\begin{align}
\Delta &\leq n^{\frac{1}{(\log n)^\varepsilon}}+1\nonumber\\
&\leq 2n^{\frac{1}{(\log n)^\varepsilon}}.\label{eq:deltat}
\end{align}

Let $R$ denote a path of maximum order in $T$. From the classic inequality $|T| \leq \Delta ^{|R|-1}$ we get
\begin{align*}
|R|-1&\geq \frac{\log |T|}{\log \Delta}\\
& \geq \frac{\log n - (\log n)^{1 - \varepsilon}}{(\log n)^{1 - \varepsilon} + 1} & \text{from \eqref{eq:tbound} and \eqref{eq:deltat}}\\
&\geq \frac{1}{2} \left ( (\log n)^\varepsilon - 1 \right ) & \text{from \eqref{eq:tribound}}\\ 
\text{so}\quad |R| &\geq \frac{1}{2} (\log n)^\varepsilon\\
&\geq (\log n)^{\varepsilon/2}&\text{from \eqref{eq:tribound}}. 
\end{align*}

As in the first part of the proof we iteratively contract the edges of $P$ that do not have both endpoints in $\bigcup_{t \in V(R)} \beta(t)$ in order to produce a graph $H$ that has a Hamiltonian path and a varied path representation $(R, \{R_v\}_{v\in V(H)})$ with adhesion less than $a$.
We can now apply Lemma~\ref{lem:adhenough} to get an induced path of order at least
\begin{align*}
\frac{1}{3} |R|^{\frac{1}{2a}} &\geq \frac{1}{3} (\log n)^{\frac{\varepsilon}{4a}}\\ 
&\geq c (\log n)^{\frac{1}{4a + \frac{1}{d}}} & \text{as}\ c\leq \frac{1}{3}.
\end{align*}

Again such a path also exists in $G$ and we are done.
\end{proof}

\subsection{Piecing things together}\label{sec:sum}

We are now ready to prove the following theorem from which will follow~Theorem~\ref{th:minclo}.
\begin{theorem}\label{th:master}
  Let $k \in \N$ and let $\mathcal{G}_k$ denote the class of graphs that either are $(k,k,k,k)$-almost embeddable or have $(k, k)$-almost bounded degree.
  There are constants $c \in \R^+$ and $d \in (0,1)$ such that if a graph $G \in \TR(\mathcal{G}_k)$ has a Hamiltonian path and order $n$, then $G$ has an induced path of order at least $c (\log n)^d$.
\end{theorem}

\begin{proof}
  Observe that a $(k,k)$-almost bounded degree graph does not contain a clique of order $2k+2$. Also there is a $k' \in \N$ such that no $(k,k,k,k)$-almost embeddable graph contains a clique of order $k'$ (see for instance \cite[Lemma~21]{DUJMOVIC2017111} for a linear upper-bound in terms of $k$).
  So for $a = \max(2k+2, k')$ we have $\TR(\mathcal{G}_k) \subseteq \TR(\mathcal{G}_k, a)$.
  By Lemma~\ref{lem:almemb} and Lemma~\ref{lem:abdd} the class $\mathcal{G}_k$ satisfies \lpp{} with the function $n \mapsto c' (\log n)^{d'}$ for some constants $c' \in \R^+$ and $d' \in (0,1)$ depending on~$k$. Also, notice that $\mathcal{G}_k$ is closed under taking subgraphs.
  Together with Lemma~\ref{lem:trga} this implies that $\TR(\mathcal{G}_k,a)$ (hence $\TR(\mathcal{G}_k)$) satisfies \lpp{} with the function $n \mapsto c(\log n)^{d}$ for some constants $c \in \R^+$ and $d \in (0,1)$ depending on~$k$.
\end{proof}

Theorem~\ref{th:minclo} is a consequence of Theorem~\ref{th:master} and the structure theorem of Grohe and Marx for graphs excluding a topological minor, stated hereafter in the setting of tree representations.

\begin{theorem}[\cite{grohe2015structure}]\label{th:structexcltopomin}
  For every graph $H$ there is an integer $k \in \N$ such that every graph not containing $H$ as a topological minor has a tree representation $(T, \{T_v\}_{v \in V(G)})$ such that for every $t \in V(T)$ the torso at $t$ is either $(k,k,k,k)$-almost embeddable or has $(k,k)$-almost bounded degree.
\end{theorem}

\section{Open problems}\label{sec:op}

A first direction for future work is to investigate how widely the results proved in this paper could be generalized. What are the most general graph classes where the bound for property \lpp{} is (at least) polylogarithmic? We recall that it was conjectured in \cite{esperet2017long} that it is the case for $k$-degenerate graphs, for every~$k\in \N$ (Conjecture~\ref{conj:esp}).
Note that the bounds we obtained in Theorem~\ref{th:pwintro} and Corollary~\ref{cor:atminas} are polynomial. An interesting task could be to characterize hereditary classes where such a property holds. Also, as mentioned in the introduction, Esperet et al. proved that in $k$-trees (i.e. edge-maximal graphs of treewidth $k$) the property \lpp{} holds with a $\Omega(\log n)$ bound, while their upper-bound for graphs of treewidth at most $k$ (recalled in the table below) shows that such a bound where $k$ does not appear in the exponent of $\log n$ does not hold for graphs of treewidth at most~$k$. This suggests that our results could be improved in the restricted setting of edge-maximal graphs from the considered classes.

Finally, a natural research direction about this problem is to obtain tight bounds for our theorems, especially in the cases of bounded pathwidth or treewidth, for which we recall below the known bounds.

\begin{center}
\begin{tabular}{lll}
 class & lower-bound & upper-bound \\ 
 \hline
 graphs of $\pw < k$ & $\frac{1}{3} n^{1/k}$\hfill (Th.~\ref{th:pwintro}) & $n^{2/k} + 1$\hfill (Th.~\ref{th:pwlb})\\[0.5em]  
 graphs of $\tw < k$ & $\frac{1}{4} (\log n)^{1/k}$\hfill (Th.~\ref{th:twintro}) & $(k+1) (\log n)^{2/(k-1)}$\quad \cite{esperet2017long}\\[0.5em]
 \end{tabular}
\end{center}

To the best of our knowledge, this question is also open for planar graphs (and more generally graphs of bounded Euler genus) with the current best lower- and upper-bounds on the function $f$ of property \lpp{}, both due to \cite{esperet2017long}, being
\[\left (\frac{1}{2\sqrt{6}} - o(1) \right ) \sqrt{\log n} \leq f(n) \leq  \frac{3 \log n}{\log \log n}.
\]

\newcommand{\etalchar}[1]{$^{#1}$}

\end{document}